\documentclass[11pt,a4paper,english]{article}

\usepackage{amsmath,amssymb,array,amsfonts,amsthm}

\usepackage[utf8]{inputenc}
\usepackage[OT1]{fontenc}
\usepackage{verbatim}
\usepackage{booktabs}
\usepackage{url}

\usepackage{floatflt}
\usepackage[algoruled,linesnumbered]{algorithm2e}
\usepackage{enumerate}

\usepackage{tikz}
\usetikzlibrary{arrows,shapes,through}

\newtheorem{Theorem}{Theorem}[section]

\theoremstyle{definition}

\newtheorem{Problem}{Problem}

\newcommand{\RFLCS}{\ensuremath{\textsc{RF-LCS}}\xspace}
\newcommand{\CLCS}{\ensuremath{\textsc{C-LCS}}\xspace}
\newcommand{\VLCS}{\ensuremath{\textsc{DC-LCS}}\xspace}
\newcommand{\SCS}{\ensuremath{\textsc{SCS}}\xspace}
\renewcommand{\emptyset}{\ensuremath{\varnothing}}

\DeclareMathOperator{\rev}{rev}

\begin{document}

\title{Variants of Constrained Longest Common Subsequence}
\author{
Paola Bonizzoni%
\thanks{Dipartimento di Informatica, Sistemistica e
    Comunicazione, Universit\`a degli Studi di Milano-Bicocca, Milano - Italy,
    {bonizzoni@disco.unimib.it}}
\and
Gianluca Della Vedova%
\thanks{Dipartimento di Statistica,
    Universit\`a degli Studi di Milano-Bicocca Milano - Italy,
    {gianluca.dellavedova@unimib.it}}
\and
Riccardo Dondi%
\thanks{Dipartimento di Scienze dei Linguaggi, della Comunicazione
  e degli Studi Culturali, Universit\`a degli Studi di Bergamo, Bergamo, Italy,
  {riccardo.dondi@unibg.it}}
\and
Yuri Pirola%
\thanks{Dipartimento di Informatica, Sistemistica e
    Comunicazione, Universit\`a degli Studi di Milano-Bicocca, Milano - Italy,
    {pirola@disco.unimib.it}}
}
\date{}

\maketitle

\begin{abstract}
In this work, we consider a variant of the classical Longest Common
Subsequence problem called Doubly-Constrained Longest Common
Subsequence (\VLCS).
Given two strings $s_1$ and $s_2$ over an alphabet $\Sigma$, a set $C_s$
of strings, and a function $C_o : \Sigma \to N$, the \VLCS problem
consists in finding the longest subsequence $s$ of $s_1$ and $s_2$ such
that $s$ is a supersequence of all the strings in $C_s$ and such
that the number of occurrences in $s$ of each symbol $\sigma \in \Sigma$
is upper bounded by $C_o(\sigma)$.
The \VLCS problem provides a clear mathematical formulation of a
sequence comparison problem in Computational Biology and generalizes two
other constrained variants of the LCS problem: the Constrained LCS and
the Repetition-Free LCS.
We present two results for the \VLCS problem.
First, we illustrate a fixed-parameter algorithm where the
parameter is the length of the solution.
Secondly, we prove a parameterized hardness result for the Constrained
LCS problem when the parameter is the number of the constraint strings
($|C_s|$) and the size of the alphabet $\Sigma$.
This hardness result also implies the parameterized hardness of the \VLCS
problem (with the same parameters) and its NP-hardness when the size of
the alphabet is constant.
\end{abstract}

\section{Introduction}
\label{introduction}

The problem of computing the longest common subsequence (LCS)
of two  sequences is a fundamental problem in stringology and in the whole
field of algorithms, as it couples a wide range of applications with a simple
mathematical formulation.
Applications of variants of LCS range from Computational Biology to data
compression, syntactic pattern recognition and file comparison (for
instance it is used in the Unix \emph{diff} command).
%%Moreover it is widely used as the first example of
%%problem that can be solved via
%%dynamic programming~\cite{CLRalg}.

A few basic definitions are in order.
Given two sequences $s$ and $t$  over a finite alphabet
$\Sigma$, $s$
is a \emph{subsequence} of $t$ if $s$ can be obtained from $t$
by removing  some (possibly zero) characters.
When $s$ is a subsequence of
$t$, then $t$ is a \emph{supersequence} of $s$.
Given two sequences $s_1$ and $s_2$, the longest common subsequence
problem asks for a longest possible sequence $t$ that is a subsequence
of both $s_1$ and $s_2$.

The problem of computing the longest common subsequence of two sequences
has been deeply
investigated and polynomial time algorithms are
well-known for the problem~\cite{PD94}.
%%, and a number of
%%algorithms have been proposed in order
%%to improve the running time for typical instances
%%\cite{AG87,Pev92,JV92,Ric95}, but all these algorithms still have a
%%$O(n^2)$
%%time complexity in the worst case. The only algorithm that has broken
%%this barrier is the one by Masek and Paterson \cite{MP80} based on the
%%Four Russians' technique \cite{4russians}; their algorithm has
%%$O(n^2/\log n)$ time complexity.
It is possible to generalize the LCS problem to a set of sequences: in such
case the result is a sequence that is a subsequence of all input sequences.
The problem  is NP-hard even on binary alphabet~\cite{Mai78} and it is not
approximable within factor $O(n^{1-\varepsilon})$, for any constant
$\varepsilon >0$, on arbitrary alphabet~\cite{JL95}.
%%A closely related problem is that of
%%computing the shortest common supersequence of a set of sequences, which is
%%NP-hard  on binary alphabet \cite{RU81}, APX-hard~\cite{Mai78} and W[$t$]-hard~\cite{BDFHW95} on
%%arbitrary alphabet.

Computational Biology is a field where several variants of the LCS
problem have been
introduced for various purposes.
For instance researchers defined some similarity measures between genome
sequences based on constrained forms of the LCS problem.
More precisely, it has been studied an LCS-like problem that
deals with two types of symbols ({\em mandatory} and {\em optional}
symbols) to model the differences in the number of occurrences allowed for
each gene~\cite{DBLP:journals/tcbb/BonizzoniVDFRV07,AdiDAM2009}.
An illustrative example is the definition of repetition-free
longest common subsequence~\cite{AdiDAM2009} where, given two sequences $s_1$
and
$s_2$, a repetition-free common subsequence
is a subsequence of both $s_1$, $s_2$ that contains at most one occurrence of
each  symbol.
Such a model can be useful in the genome rearrangement analysis, in particular
when dealing with the exemplar model.
In such framework we want to compute an exemplar sequence, that is a
sequence that contains only one representative (called the exemplar) for
each family of duplicated genes inside a genome.
In biological terms, the exemplar gene may correspond to the original
copy of the gene, from  which all other copies have been originated.

A different variant of LCS that has been introduced to compare
biological sequences is called Constrained Longest Common
Subsequence~\cite{DBLP:journals/ipl/Tsai03}.
More precisely, such variant of LCS can be useful when
comparing two biological sequences
that have a known substructure in common~\cite{DBLP:journals/ipl/Tsai03}.
Given two sequences $s_1$, $s_2$, and a constraint sequence $s_c$, we
look for a longest common subsequence $s$ of $s_1$, $s_2$, such that
$s_c$ is a subsequence of $s$.
The constrained LCS problem admits polynomial-time
algorithms~\cite{DBLP:journals/ipl/Tsai03,DBLP:journals/ijfcs/ArslanE04,DBLP:journals/ipl/ChinSFHK04}
but it becomes NP-hard when generalized to a set of input
sequences or to a set of constraint
sequences~\cite{DBLP:conf/cpm/GotthilfHL08}.

In this paper we introduce a new problem, called
\emph{Doubly-Constrained Longest Common Subsequence} and denoted as $\VLCS$, that extends
both the repetition-free longest common subsequence problem and the constrained
longest common subsequence problem.
More precisely, given two input sequences $s_1$, $s_2$,
the $\VLCS$ problem asks for the longest common subsequence $s$ that
satisfies two constraints:
(i) the number of occurrences of each symbol $\sigma$ is upper bounded
by a quantity $C_o(\sigma)$, and (ii)
$s$ is a supersequence of the strings of a specified constraint set.
First, we design a fixed-parameter
algorithm~\cite{ParameterizedComplexity} when the parameter is the length of
the solution.
Then we give a parameterized hardness result for the Constrained Longest
Common Subsequence, when the number of constraint sequences and the size
of the alphabet are considered as parameters.
This result implies the same parameterized hardness result of \VLCS.

\section{Basic Definitions}
\label{sec:basic-definitions}

Let $s_1$, $s_2$ be two strings over an alphabet $\Sigma$.
Given a string $s$, we denote by $s[i]$ the symbol at position $i$ in string $s$,
and by $s[i \dots j]$, the substring of $s$ starting at position $i$ and ending
at position $j$.
A \emph{string constraint} $C_S$ consists of a set of strings, while
an \emph{occurrence constraint} $C_o$ is a function $C_o: \Sigma \to
\mathbb{N}$, assigning an
upper bound on the number of occurrences of each symbol in $\Sigma$.
First, consider the following variant of the LCS problem.

\begin{Problem}
\label{prob:defCLCS} \textsc{Constrained Longest Common Subsequence} (\CLCS)\\
\textbf{Input:} two strings $s_1$ and $s_2$, a string constraint $C_s$.\\
\textbf{Output:} a longest common subsequence $s$ of $s_1$ and $s_2$,
so that each string in $C_s$ is a subsequence of $s$.
\end{Problem}

The problem admits a polynomial time algorithm when $C_s$ consists of a
single string~\cite{DBLP:journals/ipl/Tsai03,
  DBLP:journals/ijfcs/ArslanE04, DBLP:journals/ipl/ChinSFHK04}, while it
is NP-hard when $C_s$ consists of an arbitrary number of
strings~\cite{DBLP:conf/cpm/GotthilfHL08}.
In the latter case, notice that $\CLCS$ cannot
be approximated, since a feasible solution for the $\CLCS$ problem must
be a supersequence of all the strings in the constraint $C_s$ and
computing if such a feasible solution exists is
NP-complete~\cite{DBLP:conf/cpm/GotthilfHL08}.

\begin{Problem}
\label{prob:defRFLCS} \textsc{Repetition-free Longest Common
  Subsequence} (\RFLCS)\\
\textbf{Input:} two strings $s_1$ and $s_2$.\\
\textbf{Output:} a longest common subsequence $s$ of $s_1$ and $s_2$,
so that $s$ contains at most one occurrence of each symbol $\sigma
\in \Sigma$.
\end{Problem}

The problem is APX-hard even when each symbol occurs at most twice in
each of the input strings $s_1$ and $s_2$~\cite{AdiDAM2009}.
A positive note is that allowing at most $k$ occurrences of each symbol in
each of $s_1$ and $s_2$ results in a $\frac{1}{k}$-approximation
algorithm~\cite{AdiDAM2009}.

We can introduce an even more general version of both the $\CLCS$ and $\RFLCS$
problem, called \emph{Doubly-Constrained Longest Common Subsequence}
(\VLCS) problem.

\begin{Problem}
\label{prob:defVLCS} \textsc{Doubly-Constrained Longest Common
  Subsequence} (\VLCS)\\
\textbf{Input:} two strings $s_1$ and $s_2$, a string constraint $C_s$,
and an occurrence constraint $C_o$.\\
\textbf{Output:} a longest common subsequence $s$ of $s_1$ and $s_2$,
so that each string in $C_s$ is a subsequence of $s$ and $s$
contains at most $C_o(\sigma)$ occurrences of each symbol $\sigma \in
\Sigma$.
\end{Problem}

It is easy to see that $\CLCS$ problem is the restriction of the $\VLCS$
problem when $C_o(\sigma)=|s_1|+|s_2|$ for each $\sigma \in \Sigma$.
At the same time, the $\RFLCS$ problem is the restriction of the $\CLCS$
problem when $C_s=\emptyset$ and $C_o(\sigma)=1$ for each $\sigma \in
\Sigma$.
Therefore the $\VLCS$ problem is APX-hard, since it inherits all
hardness properties of $\CLCS$ and $\RFLCS$.

\section{A Fixed-Parameter Algorithm for \VLCS}

Initially we present a fixed-parameter algorithm for the
$\VLCS$ problem when $|C_s| \leq 1$ (hence the result holds also for the
$\RFLCS$ problem), where the parameter is the size of a solution of
$\VLCS$.
Later on, we will extend the algorithm to a generic set $C_s$.

The algorithm is based on the color coding
technique~\cite{Alon:Yuster:Zwick:1995}.
We recall the basic definition of perfect family of hash
functions~\cite{SchmidtS90}.
Given a set $S$, a family $F$ of hash functions
from $S$ to $\{1, 2, \dots , k \}$ is called \emph{perfect} if for
any $S' \subseteq S$ of size $k$, there exists an injective
hash function $f \in F$ from $S'$ to the set of labels $\{ 1, 2, \dots, k \}$.

Since $|C_s| \leq 1$, we denote by $s_c$ the only sequence in $C_s$.
Let  $k$ be the size of a solution for
$\VLCS$, and recall that a solution contains at most $C_o(\sigma)$
occurrences of each symbol $\sigma \in \Sigma$.
Notice that, since  $s$ is a subsequence of both $s_1$ and $s_2$, and by the
definition of $C_o$, the number of occurrences of each symbol $\sigma \in
\Sigma$ in
a solution $s$ is also (upper) bounded by the number of occurrences of
$\sigma$ in each $s_1$ and $s_2$ (i.e.~$occ(\sigma, s) \leq \min
\{ C_o(\sigma), occ(\sigma, s_1), occ(\sigma, s_2) \}$).
Let $C'_o$ be a function from $\Sigma$ to $\mathbb{N}$ defined as $C'_o(\sigma)
:= \min \{ C_o(\sigma), occ(\sigma, s_1), occ(\sigma, s_2) \}$.

Given $C'_o$ and the sequences $s_1$ and $s_2$, we construct a set
$\widetilde{\Sigma}$ that contains the pairs $(\sigma, i)$ for each
$\sigma \in \Sigma$ and $i \in \{1, \dots, C'_o(\sigma)\}$.
For example, if
$s_1= aaaa\,bbb\,cc\,d$, $s_2=dd\,c\,bbbb\,aaaa$, and
$C_o(a) = C_o(b) = C_o(c) = C_o(d) = 3$,  then the set
$\widetilde{\Sigma}$ is equal to $\{ (a,1), (a,2), (a,3), (b,1), (b,2),
(b,3), (c,1), (d,1) \}$.

Consider now a perfect family $F$ of hash functions from $\widetilde{\Sigma}$
to the set $\{ 1, 2\, \dots, k \}$.
We can associate a function $l:\Sigma\to 2^{\{ 1, 2\, \dots, k \}}$
with each $f\in F$, where
$l(\sigma)=\{ f(\sigma, i): (\sigma, i) \in \widetilde{\Sigma} \}$.
Let $s$ be a solution of the $\VLCS$ problem of length at most $k$, and
let $L$ be a subset of $\{1, \ldots, k\}$.
Then $s$ is an $L$-\emph{colorful} solution w.r.t.~a hash
function $f \in F$ (and its associated function $l$) if and only if
there exists a function $l_1:\Sigma\to 2^{\{ 1, 2, \dots, k \}}$ which
satisfies the following conditions:
\begin{enumerate}[(i)]
\item $\forall \sigma\in\Sigma$, $l_1(\sigma)\subseteq l(\sigma) \cap L$,
\item $\forall \sigma\in\Sigma$, $|l_1(\sigma)|$ is equal to the number
  of occurrences of $\sigma$ in $s$,
\item $\forall\sigma_1, \sigma_2\in\Sigma$, $l_1(\sigma_1)\cap
  l_1(\sigma_2)= \emptyset$.
\end{enumerate}
% Similarly we say that $s$ is an \emph{$L$-colorful} solution w.r.t.~a
% hash function $f \in F$ (and its associated function $l$), and a subset
% $L\subseteq \{1,\dots ,k\}$ if and only if there exists a function
% $l_1:\Sigma\to 2^{\{ 1, 2\, \dots, k \}}$ such that (i) $\forall
% \sigma\in\Sigma$ $l_1(\sigma)\subseteq l(\sigma)\cap L$, (ii) $\forall
% \sigma\in\Sigma$ $|l_1(\sigma)|$ is equal to the number of occurrences of
% $\sigma$ in $s$, (iii)
% $\forall\sigma_1,
% \sigma_2\in\Sigma$
% $l_1(\sigma_1)\cap l_1(\sigma_2)=\emptyset$.
Intuitively, an $L$-colorful solution $s$ is a sequence such that it is
possible to associate distinct elements (labels) of the set $L$ with all
the characters of $s$ by using the function $l$.
Notice that the length of an $L$-colorful solution $s$ is equal to the
number of labels that $s$ uses, and each symbol $\sigma$ does not occur
more than $C'_o(\sigma)$ times in $s$.

The basic idea of our algorithm is to verify if there exists an $L$-colorful
solution that uses all labels in $L$ or, equivalently, if the length of an
optimal $L$-colorful solution is $|L|$. Such task is fulfilled via
a dynamic
programming recurrence.
Since $F$ is a perfect family of hash functions, for each
feasible solution $s$ of length $k$, there exists a hash
function $f \in F$ such that $s$ is $\{1, \ldots ,k\}$-colorful w.r.t.~$f$.
Therefore, by computing the recurrence for all hash functions of
$F$, we are guaranteed to find a solution of length $k$, if such a
solution exists.

% Given the sequence $s_t$, for each symbol $\sigma_i$, let $l(\sigma_i)$
% be the set of elements in  $\{ 1, 2\, \dots, k \}$ that are image
% through $f$ of the occurrences of $\sigma_i$ in $s_t$.
% Observe that $|l(\sigma_i)| \leq \min(o(\sigma_i), k)$.

Given a hash function $f$, we define $V[i, j, h, L]$ which takes value $1$ if
and only if there exists an $L$-colorful common subsequence $s$ of $s_1[1 \dots i]$ and $s_2[ 1 \dots
j]$, such that $s$ is a supersequence of $s_c[1 \dots h]$ and $s$ has length
equal to $|L|$ (or, equivalently, $s$ uses all labels in $L$).
Notice that the actual
supersequence can be computed by a standard backtracking technique.
Theorem~\ref{par_alg:VLSC:correct} states that $V[i, j, h, L]$ can be
computed by the following dynamic programming recurrence which is an extension
of the standard equation for the Longest Common Subsequence (LCS) problem~\cite{CLRalg}.

\begin{equation}
\label{eq-rec-VLCS} V[i ,j ,h, L ]= \max \left\{\\
\begin{array}{ll}
V[i-1 ,j , h, L ] & \\
V[i, j-1 , h , L ] & \\
V[i-1 ,j-1 , h, L\setminus \{ \lambda \} ] & \text{if } s_1[i]=s_2[j] \text{ and }
\\
 & \lambda \in L \cap l(s_1[i])\\
V[i-1 ,j-1 , h-1 , L\setminus\{ \lambda \} ] & \text{if } s_1[i]=s_2[j]=s_c[h] \text{
and }\\
 & \lambda \in L \cap l(s_1[i])
\end{array} \right.
\end{equation}

The boundary conditions are $V[0, j, h, L]= 0$ and $V[i,0 , h, L ]=0$  if  $L
\neq \emptyset$, while
$V[i,j , 0, \emptyset ]=1$, and $V[i, j, h, \emptyset]= 0$ when $h>0$.
Moreover, notice that, as a consequence of the recurrence's definition, we
have $V[i, j, h, L]= 0$ for all $h > |L|$.
A feasible solution of length $k$ is $\{ 1, \dots, k\}$-colorful
w.r.t.~$f$ if and only if
$V[|s_1| , |s_2| , |s_c|, \{ 1, \dots, k   \} ]=1$.
In this case, a standard backtracking search can reconstruct the actual
solution.

\begin{Theorem}
\label{par_alg:VLSC:correct}
Let $f\in F$ be a hash function mapping injectively the solution $s$ to the
set of labels $\{1,\ldots, k\}$.
Then Equation~(\ref{eq-rec-VLCS}) is correct.
\end{Theorem}
\begin{proof}
We will prove the theorem by induction, that is we will  prove the correctness of the
value in $V[i_a,j_a , h_a, L_a ]$ by assuming that of
$V[i_b,j_b , h_b, L_b ]$ when  $i_b \leq i_a$,  $j_b \leq i_a$, $h_b \leq
h_a$, $L_b \subseteq L_a$, and at least one inequality is strict.

Let $s$ be an optimal $L_a$-colorful solution for the sequences
$s_1[1,\ldots,i_a]$, $s_2[1,\ldots,j_a]$, $s_c[1,\ldots,h_a]$,
and let $\beta$ be the last symbol of $s$, that is
$s=t\beta$, where $t$ is the prefix of $s$ consisting of all but
the last character.

If $\alpha\neq \beta$ then, just as for the recurrences of the standard LCS
problem~\cite{CLRalg},
the theorem holds.

Therefore we can assume now that $\alpha=\beta$.
Since $s$ is $L_a$-colorful, then there exists a mapping $l_1$ satisfying the
definition of $L$-colorfulness. By condition (ii), $|l_1(\beta)|$ is equal to
the number of occurrences of $\beta$ in $s$. Let $z$ be the label which is
image through $f$ of the last character of $s$.
Then there exists an $L\setminus\{z\}$-colorful solution $t$ of $s_1[1,\ldots,
i_a-1]$, $s_2[1,\ldots, j_a-1]$, $s_c[1,\ldots, j_a]$ (if $t$ is a
supersequence of $s_c[1,\ldots, j_a]$) or of $s_1[1,\ldots,
i_a-1]$, $s_2[1,\ldots, j_a-1]$, $s_c[1,\ldots, j_a-1]$, hence
completing the proof.
\end{proof}

If $f$ is a hash  function that does not map injectively the solution $s$ of length $k$
to the set of labels $\{1,\ldots, k\}$ then, by definition of hash function, there is
a label $z \in \{1,\ldots, k\}$ that is not in the image through $f$
of any character of $s$.
The latter observation also implies that $z$ is not in the image through $l$
of any symbol, therefore for each  set $L$ including $z$, the last two cases
of our recurrence equation cannot apply, which implies that
$V[i,j,h,\{1,\ldots ,k\}]=0$ for all values of $i$, $j$, $h$, hence
estabilshing the correctness of our algorithm.

It is immediate to notice that the total number of entries of the matrix
$V[\cdot,\cdot,\cdot,\cdot]$ is $|s_1||s_2||s_c|2^k$. Furthermore notice
that computing each entry requires at most $O(k)$ time, as case 1 and case 2
of the recurrence require constant time, while case 2 and case 4 require at most
$O(k)$ time, since $|L| \leq k$.
Since there exists a perfect family of hash functions whose size is
$O(\log |\widetilde{\Sigma}|)2^{O(k)}$ and that can be computed in $O( |\widetilde{\Sigma}|\log
|\widetilde{\Sigma}|)2^{O(k)}$ time~\cite{Alon:Yuster:Zwick:1995},
and $|\widetilde{\Sigma}|\leq |s_1$, the algorithm has an
overall $O (|s_1|\log |s_1|)2^{O(k)} + O(|s_1||s_2||s_c|k 2^k )$ time
complexity.

The algorithm actually computes a longest supersequence of $s_c$ that is a
feasible solution of the problem. Assume now that $C_s$ is a generic
occurrence set, and let $x$ be an optimal solution of a generic instance of
the \VLCS problem of size $k$.
It is immediate to notice that, by removing from $x$ all
symbols that are not also in one of the sequence of $C_s$, we obtain a common
supersequence $x_1$ of $C_s$ that is a subsequence of $x$. Moreover,
as $x$ has size $k$, $x_1$
contains at most $k$ characters (where $k$ is the length of an optimal
solution).

Notice that the alphabet consisting of the symbols appearing in at least one
sequence of $C_s$ contains at most $k$ symbols, for otherwise all supersequences of
$C_s$ would be longer than $k$. Consequently there are at most $k^k$ such
supersequences. Our algorithm for a generic $C_s$ enumerates all such
supersequences $s_c$, and applies the algorithm for $|C_s|=1$ on the new set
of constraint sequences made only of $s_c$, returning the longest feasible
solution computed.

The overall time complexity is clearly
$O \left( k^k (|s_1|\log |s_1|)2^{O(k)} + |s_1||s_2||s_c|k 2^k ) \right)$.

\section{W[$1$]-hardness of $\CLCS$}

In this section we prove that computing if there exists a feasible solution
of $\CLCS$ is not only NP-complete, but it is also W[$1$]-hard
when the parameter is the number of string in $C_s$ and the alphabet
$\Sigma$ (see \cite{ParameterizedComplexity} for an exposition on the
consequences of  W[$1$]-hardness).

We reduce the  Shortest Common Supersequence ($\SCS$) problem parameterized by the
number of input strings and the size of alphabet $\Sigma$, which is known
to be W[$1$]-hard~\cite{DBLP:journals/jcss/Pietrzak03}.
Let $R=\{r_1, \ldots, r_k\}$ be a set of sequences over alphabet $\Sigma$, hence $R$ is a generic
instance of the \SCS problem.
In what follows we denote by $l$ the size of a solution of the $\SCS$ problem.

The input of the $\CLCS$ consists of two sequences $s_1$, $s_2$, and a
string constraint $C_s$.
Let $\#$ be a delimiter symbol not in $\Sigma$. Moreover, given a sequence
$r_i=y_1y_2\cdots y_z$ over alphabet $\Sigma$, let $c(r_i)$ be the sequence
$y_1\#y_2\#\cdots \#y_z\#$.
Pose $C_s=\{\#^l\} \cup \{c(r_i): r_i\in R\}$,
%  and notice that, since any feasible solution must be a
% supersequence of the set
% $C_s=\{r_1, \ldots, r_k\}$, for our purposes it suffices to find, for any given set $C_s$,  two
% strings $s_1$ and $s_2$ such that no  common subsequence
% of $s_1$ and $s_2$ can have length larger than a certain parameter $l$, and
% there exists a common subsequence set of common supersequences of $C_s$ shares only
% sequences of the same length (more precisely, the shared sequences are the
% shortest common supersequences of $C_s$).
let $w$ be a sequence over $\Sigma$ such that
$w$ contains  exactly one occurrence of each symbol in $\Sigma$, and let
$\rev(w)$ be the reversal of $w$. Finally, let $s_1=(w\#)^l$
and $s_2=(\rev(w)\#)^l$. In the following we call each occurrence of $w$ or of
$\rev(w)$ a \emph{block}.

Let $t$ be any supersequence of $\#^l$ that is also a common subsequence of
$s_1$ and $s_2$.
Since in each of those sequences there are $l$ $\#$s, then also $t$ must
contain $l$ $\#$s, which in turn implies that, by construction of $w$, at most one symbol
of each block can be in $t$.  Therefore $t$ contains at
most $2l$ symbols.
At the same time, let $p$ be a generic sequence  no longer than $2l$,
ending with a $\#$ and such that no two symbols from $\Sigma$ appear
consecutively in $p$.
Since each symbol of $\Sigma$ occurs exactly once in $w$,  it is immediate to
notice that $p$ is a common  subsequence of $s_1$ and $s_2$.
Consequently, the set of all supersequences of $\#^l$ that are also  common subsequences of
$s_1$ and $s_2$ is equal to the set of sequences $q$ with length not larger
than $2l$ and such that (i) $q$ contains exactly $l$ $\#$s, (ii) $q$ ends
with a $\#$, and (iii) taken two consecutive symbols from $q$, at least one of
those symbols is equal to $\#$.

An immediate consequence is that there
exists a feasible solution of length $2l$ of the instance of \CLCS made of
the set $C_s$ and
the two sequences $s_1$ and $s_2$ iff there exists a supersequence of length
$2l$ of the set $R$ of sequences.

The reduction described is an
FPT-reduction~\cite{ParameterizedComplexity}.
Finally, notice that the W[1]-hardness of $\CLCS$ with parameters
$|C_s|$ and $|\Sigma|$ implies the W[1]-hardness of $\VLCS$ with
parameters $|C_s|$ and $|\Sigma|$ since $\CLCS$ is a restriction of the
$\VLCS$ problem.

Moreover, notice that the same reduction can be applied starting from the
SCS problem over binary alphabet, implying that
the \VLCS problem is NP-hard
over a fixed ternary alphabet, as the \SCS problem is NP-hard over a binary
alphabet~\cite{RU81}.

\end{document}